\documentclass[11pt,bezier]{article}
\usepackage{amsmath,amssymb,amsfonts,graphics}

\newtheorem{prethm}{{\bf Theorem}}

\newenvironment{thm}{\begin{prethm}{\hspace{-0.5
               em}{\bf.}}}{\end{prethm}}

\newtheorem{prepro}{{\bf Theorem}}%
\newtheorem{precor}{{\bf Corollary}}

\newenvironment{cor}{\begin{precor}{\hspace{-0.5
               em}{\bf.}}}{\end{precor}}
\newtheorem{preconj}{{\bf Conjecture}}

\newtheorem{preremark}{{\bf Remark}}

\newenvironment{remark}{\begin{preremark}\rm{\hspace{-0.5
               em}{\bf.}}}{\end{preremark}}
\newtheorem{prelem}{{\bf Lemma}}

\newenvironment{lem}{\begin{prelem}{\hspace{-0.5
               em}{\bf.}}}{\end{prelem}}
\newtheorem{preproof}{{\bf Proof.}}

\newenvironment{proof}[1]{\begin{preproof}{\rm
               #1}\hfill{$\Box$}}{\end{preproof}}

\def\emline#1#2#3#4#5#6{%
       \put(#1,#2){\special{em:moveto}}%
       \put(#4,#5){\special{em:lineto}}}
\def\newpic#1{}

\title{\large \bf On The Signed Edge Domination\\ Number of
Graphs
\thanks {2000{ \it Mathematics Subject Classification}: 05C69, 05C78.
}}

\author{\bf\small\sc S. Akbari, S. Bolouki, P. Hatami, M. Siami}

\date{}
\begin{document}
\maketitle
\begin{abstract}

Let $\gamma'_s(G)$   be the signed edge domination number of G. In
2006, Xu conjectured that: for any $2$-connected graph G of order
$ n (n \geq 2),$  $\gamma'_s(G)\geq 1$. In this article we show
that this conjecture is not true. More precisely, we show that
for any positive integer $m$, there exists an $m$-connected graph
$G$ such that $ \gamma'_s(G)\leq -\frac{m}{6}|V(G)|.$ Also for
every two natural numbers $m$ and $n$, we determine
$\gamma'_s(K_{m,n})$, where $K_{m,n}$ is the complete bipartite
graph with part sizes $m$ and $n$.

\end{abstract}
\vspace{1cm} \noindent{\Large\sc  Introduction} \vspace{4mm}

In this paper all of graphs that we consider are finite, simple and
undirected. Let $G=(V(G),E(G))$ be a graph with vertex set $V(G)$
and edge set $E(G)$. The {\it order} of $G$ denotes the number of
vertices of $G$. For any $v\in V(G)$, $d(v)$ is the degree of $v$
and $E(v)$ is the set of all edges incident with $v$. If $e=uv\in
E(G)$, then we put $N[e]=\{u'v'\in E(G)|u'=u$ or $v'=v\}$. Let $G$
be a graph and $f: E(G)\longrightarrow \{-1,1\}$ be a function. For
every vertex $v$, we define $s_v= \sum_{e\in E(v)} f(e)$. We denote
the complete bipartite graph with two parts of sizes $m$ and $n$, by
$K_{m,n}$.  Also we denote the cycle of order $n$, by $C_n$. In
\cite{xu3} the signed edge domination function of graphs was
introduced as follows:

Let $G=(V(G),E(G))$ be a non-empty graph. A function
$f:E(G)\longrightarrow \{-1,1\}$ is called a {\it signed edge
domination function} {\it {\rm (SEDF)}} of $G$ if $\sum_{e'\in
N[e]}f(e')\geq 1$, for every $e\in E(G)$. The {\it signed edge
domination number} of $G$ is defined as,
$$\gamma'_s(G)=min\{\sum_{e\in E(G)}f(e) \,| \,\,f {\rm \,\,is \,\,an\,}
\,{\rm SEDF} \,\,{\rm of}\,\, G\}.$$ Several papers have been
published on lower bounds and upper bounds of the signed edge
domination number of graphs, for instance, see \cite{xu1},
\cite{xu2}, \cite{xu3}, \cite{ze}, \cite{zh}. In \cite{xu1}, Xu
posed the following conjecture:

\noindent For any $2$-connected graph G of order $ n (n \geq 2),$
$\gamma'_s(G)\geq 1$.

In the first section we give some counterexamples to this
conjecture by showing that for any natural number $m$,  there
exists an $m$-connected graph $G$ such that $ \gamma'_s(G)\leq
-\frac{m}{6}|V(G)|.$ For any natural number $k$, let
$g(k)=min\{\gamma'_s(G)\,|\, |V(G)|=k\}$. In \cite{xu1} the
following problem was posed:

\noindent Determine the exact value of $g(k)$ for every positive
integer $k$. In Section 1, it is shown that for every natural
number $k$, $k\geq 12$, $g(k)\leq \frac{-(k-8)^2}{72}$.

\vspace{1cm} \noindent{\Large\sc 1. Counterexamples to a
Conjecture} \vspace{4mm}

In this section we present some counterexamples to a conjecture that
appeared in $\cite{xu1}$. We start this section by the following
simple lemma and leave the proof to the reader.

\begin{lem}\label{edge} Let $f:E(G)\longrightarrow \{-1,1\}$ be
a function.  Then $f$ is an {\rm SEDF} of $G$, if and only if  for
any edge $e=uv$, $s_u+s_v-f(e) \geq 1.$ Moreover, if $f$ is an
{\rm SEDF}, then $s_u+s_v\geq 0$.
\end{lem}

An $L_{(m,n)}$-graph $G$ is a graph of order $(n+1)(mn+m+1)$,
whose vertices can be partitioned into  $n+1$ subsets $V_1,
\ldots, V_{n+1}$ such that:

\noindent (i) The induced subgraph on $V_1$ is the complete graph
$K_{mn+m+1}$.

\noindent (ii) The induced subgraph on $V_i$, $2\leq i\leq n+1$ is
the complement of $K_{mn+m+1}$.

\noindent (iii) For every $i$, $2\leq i\leq n+1$, all edges
between $V_1$ and $V_i$ form $m$ disjoint matchings of size
$mn+m+1$.

\noindent (iv) There is no edge between $V_i$ and $V_j$ for any
$i, j$, $2\leq i<j\leq n+1$.

  It is well-known
that for any natural number $r$, the edge chromatic number of
$K_{r,r}$ is $r$, see Theorem 6 of \cite[p.93]{bon}. Thus for
every pair of natural numbers $m$ and $n$, there is an
$L_{(m,n)}$-graph.

\begin{thm}
\label{thm:1} Let $m$ and $n$ be two natural numbers. Then for
every $L_{(m,n)}$-graph $G$, we have,
$$
\gamma'_s(G) \leq \frac{(mn+m+1)(m-mn)}{2}.$$
\end{thm}
\begin{proof}{
To prove the inequality we provide an SEDF for $G$, say $f$,  such
that,
\begin{equation*}
\sum_{e\in E(G)} f(e)= \frac{(mn+m+1)(m-mn)}{2}.
\end{equation*}
Define $f(e)=1$, if both end points of $e$ are contained in
$V_1$, and $f(e)=-1$, otherwise. We find,

\begin{equation*}\begin{split}
\sum_{e\in E(G)} f(e) &= \frac{(mn+m+1)(mn+m)}{2}-(mn+m+1)mn\\&=
\frac {(mn+m+1)(m-mn)}{2}.
\end{split}\end{equation*}
It can be easily verified that for every $v\in V_1$, $s_v=m$, and
for every
 $v\in V(G)\setminus V_1$, $s_v=-m$. Now, Lemma \ref{edge} yields that $f$ is an SEDF for $G$.}
\end{proof}

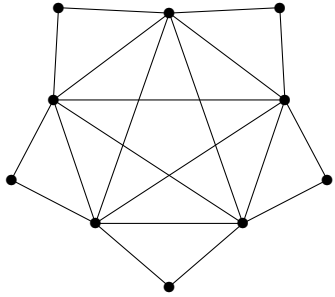
\begin{figure}[ht]\label{fig1}
\begin{center}
\unitlength=.70mm \special{em:linewidth 0.4pt}
\linethickness{0.4pt}

\begin{picture}(0,50)(0,-15)
\put(22,17.5){\circle*{2}}
\put(-22,17.5){\circle*{2}}
\put(0,34){\circle*{2}}
\put(14,-6){\circle*{2}}
\put(-14,-6){\circle*{2}}
\put(21,35){\circle*{2}}
\put(-21,35){\circle*{2}}
\put(30,2.3){\circle*{2}}
\put(-30,2.3){\circle*{2}}
\put(0,-18){\circle*{2}}
\emline{21}{35}{1}{22}{17.5}{2}
\emline{21}{35}{1}{0}{34}{2}
\emline{-21}{35}{1}{-22}{17.5}{2}
\emline{-21}{35}{1}{0}{34}{2}
\emline{30}{2.3}{1}{22}{17.5}{2}
\emline{30}{2.3}{1}{14}{-6}{2}
\emline{-30}{2.3}{1}{-14}{-6}{2}
\emline{-30}{2.3}{1}{-22}{17.5}{2}
\emline{0}{-18}{1}{14}{-6}{2}
\emline{0}{-18}{1}{-14}{-6}{2}
\emline{22}{17.5}{1}{-22}{17.5}{2}
\emline{22}{17.5}{1}{0}{34}{2}
\emline{22}{17.5}{1}{14}{-6}{2}
\emline{22}{17.5}{1}{-14}{-6}{2}
\emline{-22}{17.5}{1}{0}{34}{2}
\emline{-22}{17.5}{1}{14}{-6}{2}
\emline{-22}{17.5}{1}{-14}{-6}{2}
\emline{0}{34}{1}{14}{-6}{2}
\emline{0}{34}{1}{-14}{-6}{2}
\emline{-14}{-6}{1}{14}{-6}{2}

\end{picture}
\caption{A $2$-connected $L_{(2,1)}$-graph  with $\gamma'_s<1$.}
\end{center}
\end{figure}

\noindent {\bf Example 1}. Consider the $L_{(2,1)}$-graph $G$ shown
in Figure 1. The graph clearly has perfect matching; and by applying
Lemma~\ref{edge} to the edges of this matching we may conclude that
for every SEDF $f$ of this graph, $\sum_{e\in
E(G)}f(e)=\frac{1}{2}\sum_{v\in V(G)}s_v\geq 0$ , hence
$\gamma'_s(G)\geq 0$. But it follows from Theorem \ref{thm:1} that
$\gamma'_s(G)\leq 0$. Consequently, $\gamma'_s(G)=0$ and the bound
in Theorem \ref{thm:1} is sharp for this graph.


In \cite{xu1}, Xu conjectured that for any $2$-connected graph G
of order $ n (n \geq 2),$  $\gamma'_s(G)\geq 1$. The next theorem
shows that conjecture fails.

\begin{thm}\label{thm:2}
For any natural number $m$,  there exists an $m$-connected graph
$G$ such that \label{eq:1} $\gamma'_s(G)\leq -\frac{m}{6}|V(G)|.$
\end{thm}
\begin{proof}{
First we claim that for each pair of natural numbers $m$ and $n$,
every $L_{(m,n)}$-graph is an $m$-connected graph. To see this we
note that if one omits at most $m-1$ vertices of an
$L_{(m,n)}$-graph, then some vertices of $V_1$ remain (because
$|V_1|=mn+m+1$) and since the degree of each vertex of $V_i$, $2\leq
i\leq n+1$ is $m$, the claim is proved.

Now, for any natural number $m$, consider an $L_{(m,2)}$-graph
$G$. By Theorem~\ref{thm:1}, the following inequality holds:$$
\gamma'_s(G) \leq \frac{1}{2}(2m+m+1)(m-2m)=
-\frac{m}{6}|V(G)|.$$}
\end{proof}

\begin{remark}
If we repeat the previous proof for an $L_{(m,n)}$-graph instead
of an $L_{(m,2)}$-graph, then we find $\gamma'_s(G)\leq
\frac{-m(n-1)}{2(n+1)}|V(G)|$. Hence for large enough $n$,
$\gamma'_s(G)\leq \frac{-m+1}{2}|V(G)|$.
\end{remark}



\begin{lem}\label{1/2} Let $G$ be a graph with an {\rm SEDF}. If $G$
contains $C_n$ as subgraph, then
$$\sum_{v\in V(C_n)} s_{v} \geq
0.$$
\end{lem}

\begin{proof}{ Let $V(C_n)=\{v_1, \ldots ,v_n\}$.
Clearly, we have,
$$\sum_{i=1}^n s_{v_i} =\frac{1}{2} \sum_{i=1}^n (s_{v_i}+s_{v_{i+1}}),$$
where indices are modulo $n$. Thus by Lemma~\ref{edge}, the proof
is complete.}
\end{proof}

\begin{thm}\label{thm:lowerbound}
For every graph $G$ of order $n$, $\gamma'_s(G) \geq
\frac{-n^2}{16}.$
\end{thm}
\begin{proof}{
An elementary graph is a graph in which each component is a
$1$-regular graph or a $2$-regular graph. Let $H$ be an elementary
subgraph of $G$ with maximum number of vertices. With no loss of
generality we may assume that $H$ has no even cycle, since one can
replace an even cycle of size $2k$ by $k$ vertex-disjoint edges.
Suppose $\alpha$ is the number of vertices of $G$ which are not
covered by $H$. \noindent We claim that for every vertex $v$ which
is not covered by $H$, $d(v)\leq \frac{n-\alpha}{2}.$

To see this, we note that $v$ is  adjacent to none of the other
$\alpha-1$ vertices which are not covered by $H$, because
otherwise we could find an elementary  subgraph $H'$ which covers
more vertices of $G$, a contradiction. Also, $v$ is adjacent to
none of the vertices of an odd cycle of $H$, because  if $v$ is
adjacent to a vertex $u$ of an odd cycle $C$, we can decompose
the set $E(C)\bigcup \{uv\}$ into vertex-disjoint edges which
cover $V(C) \bigcup \{v\}$,  obtaining  an elementary subgraph
$H'$ which covers more vertices, a contradiction. If $v$ is
adjacent to both end points of an edge in the matching part of
$H$, then we can add an odd cycle of length $3$ to $H$, obtaining
a bigger elementary subgraph, a contradiction. Thus the degree of
$v$ does not exceed the number of the edges in the matching part
of $H$, so, $d(v)\leq \frac{n-\alpha}{2}.$

By Lemmas \ref{edge} and \ref{1/2}, $\sum_{v\in V(H)}s_v\geq 0$.
Therefore we have,

\begin{equation*}\begin{split}
\sum_{e\in E(G)} f(e) &= \frac{1}{2} (\sum_{v\in V(H)}s_v+ \sum_{v
\in V(G)\setminus V(H)}s_v) \geq \frac{1}{2} \sum_{v \in
V(G)\setminus V(H)}s_v\\ &\geq \frac{-1}{2}\sum_{v\in
V(G)\backslash V(H)}d(v)\geq \frac{-\alpha(n-\alpha)}{4} \geq
\frac{-n^2}{16}.
\end{split}\end{equation*}
}\end{proof}
\begin{cor} If $G$ has a spanning elementary subgraph, then $\gamma'_s(G)\geq
0$.
\end{cor}
\begin{proof}{
In the proof of the previous theorem replace $\alpha$ by $0$.}
\end{proof}

\noindent  In \cite{xu1} the following problem has been posed:

Determine the exact value of $g(k)$ for every positive integer
$k$. In the next theorem we find a lower and an upper bound for
$g(k)$, $k\geq 12$.

\begin{thm}
For every natural number $k$, $k\geq 12$, $-\frac{k^2}{16}\leq
g(k)\leq -\frac{(k-8)^2}{72}$.
\end{thm}
\begin{proof}{ The lower bound is an immediate consequence of
Theorem \ref{thm:lowerbound}. First we obtain the upper bound for
$k=9m+3$. In the proof of the Theorem~\ref{thm:1}, we constructed
a graph $G$ of order $(n+1)(mn+m+1)$ vertices for which,
$$\gamma'_s(G)\leq \frac
{(mn+m+1)(m-mn)}{2}.$$ Assume that $n=2$. We have,
$$g(9m+3)\leq \frac{-m}{6}(9m+3).$$
Since $k\geq 12$, for $k=9m+3$ we find,
\begin{equation*}
g(k)\leq \frac{-\left(\frac{k-3}{9}\right)}{6}k\leq
\frac{-k^2}{72}.
\end{equation*}

Now, for every $k$, we may write $k=9m+3+r$, where $0\leq r<9$. By
adding $r$ isolated vertices to the constructed graph for $9m+3$,
and using the previous inequality for $g(9m+3)$, we have the
following:
$$ g(k)\leq \frac{-(k-r)^2}{72}\leq
\frac{-(k-8)^2}{72},$$

and the proof is complete.}
\end{proof}

\vspace{1cm} \noindent{\Large\sc 2. Signed Edge Domination of\\
Complete Bipartite Graphs} \vspace{4mm}

\noindent In this section we want to obtain the signed edge
domination number of complete bipartite graphs.

\begin{thm}
Let $m$ and $n$ be two natural numbers where $m\leq n$. Then the
following hold:

\noindent (i) If $m$ and $n$  are even, then
$\gamma'_s(K_{m,n})=\min(2m,n)$,

\noindent (ii) If $m$ and $n$  are odd, then
$\gamma'_s(K_{m,n})=\min(2m-1, n)$,

\noindent (iii) If $m$ is even and $n$ is odd, then
$\gamma'_s(K_{m,n})=\min(3m, \max(2m, n+1))$,

\noindent (iv) If $m$ is odd and $n$ is even, then
$\gamma'_s(K_{m,n})=\min(3m-1, \max(2m, n))$.
\end{thm}

\begin{proof}{
Let $(X, Y)$ be two parts of the complete bipartite graph
$K_{m,n}$ and $X=\{u_1, \ldots ,u_m\}$ and $Y=\{v_1, \ldots ,
v_n\}$. We note that if $f$ is an SEDF for $K_{m,n}$, then we
have,
$$\sum_{e\in E(K_{m,n})}f(e)=\sum_{u\in X}s_u=\sum_{v\in Y}s_v.$$

 {\bf (i)} First we show that $\gamma'_s(K_{m,n})\geq
\min(2m,n)$. It suffices to show that if $f$ is an SEDF such that
$\sum_{e\in E(K_{m,n})}f(e)<2m$, then $\sum_{e\in
E(K_{m,n})}f(e)\geq n$. Since $\sum_{e\in E(K_{m,n})}f(e)<2m$,
there exists a vertex $u\in X$ such that $s_u<2$.  But $s_u$ is
even and so $s_u\leq 0$. If $s_u=0$, then $u$ is incident with
$n/2$ edges with value $1$ and $n/2$ edges with value  $-1$. If
$f(uv)=1$, for some  $v\in Y$, then by Lemma \ref{edge}, $s_v\geq
2$.
 If $f(uv)=-1$, for some $v\in Y$, then we find $s_v\geq 0$. Thus we have
$\sum_{e\in E(K_{m,n})}f(e)=\sum_{v\in Y}s_v\geq 2\left
(\frac{n}{2}\right)=n$. If $s_u<0$, then $s_u\leq -2$.  Now, for
each $v\in Y$, by Lemma 1, $s_v\geq 2$. Therefore we have the
following:
$$\sum_{e\in E(K_{m,n})}f(e)=\sum_{v\in Y}s_v\geq 2n>n.$$
Hence  $\gamma'_s(K_{m,n})\geq \min(2m,n)$.

We now show that there exist two SEDF, say $f$ and $g$, such that
$\sum_{e\in E(K_{m,n})}f(e)=2m$ and $\sum_{e\in E(K_{m,n})}g(e)=n$.
Let $f$ be define as follows:

\ \ $f(u_iv_j)=
\begin{cases}
1 &\ {\rm if}\,\, i+j\,\, {\rm is \,\, odd}\\
1 &\ {\rm if}\,\, i=j \\
-1  &\ {\rm otherwise.}
 \end{cases}$

It is clear that for every $u_i$, $s_{u_i}=2$. Also one can see that
$s_{v_i}\geq 0$, for $i=1,\ldots ,n.$ Now, by Lemma~\ref{edge}, we
see that $f$ is an SEDF. Therefore,
$$\gamma'_s(K_{m,n})\leq\sum_{e\in E(K_{m,n})}f(e)=\sum_{u\in
X}s_u=2m,$$
as required.

Define  $g$ as follows:

\ \ $g(u_iv_j)=
\begin{cases}
1  &\  {\rm if}\,\, i+j\,\, {\rm is\,\,  odd}\\
1   &\  {\rm if}\,\, i\,\, {\rm is\,\, even\,\, and}\,\, i=j\,\,
{\rm
modulo}\,\, m\\
-1 &\ {\rm otherwise}.
 \end{cases}$

 We note that if $i$ is even, then
$s_{v_i}=2$; and if $i$ is odd, then $s_{v_i}=0$. Also, if $i$ is
even, then $s_{u_i}\geq 2$; and if $i$ is odd, then $s_{u_i}=0.$
Now, Lemma~\ref{edge} implies that $g$ is an SEDF. Therefore,
$$\gamma'_s(K_{m,n})\leq\sum_{e\in E(K_{m,n})}g(e)=\sum_{i=1}^ns_{v_i}=\frac{2n}{2}=n,$$
as required.

 {\bf (ii)} First we show that $\gamma'_s(K_{m,n})\geq
\min(2m-1, n)$. It is enough to show that if $f$ is an SEDF with
$\sum_{e\in E(K_{m,n})}f(e)<n$, then $\sum_{e\in E(K_{m,n})}f(e)\geq
2m-1$. Since $\sum_{e\in E(K_{m,n})}f(e)<n$, there exists a vertex
$v\in Y$ such that $s_v<1$. But $s_v$ is odd and so $s_v\leq -1$. If
$s_v=-1$, then $v$ is incident with $\frac{m-1}{2}$ edges with value
$1$ and $\frac{m+1}{2}$ edges with value $-1$. If $f(uv)=1$, for
some $u\in X$, then by Lemma 1, $s_u\geq 3$. If $f(uv)=-1$, for some
$u\in X$, then similarly we have $s_u\geq 1$. Thus we have the
following:
$$\sum_{e\in E(K_{m,n})}f(e)=\sum_{u\in X}s_u\geq
3\left(\frac{m-1}{2}\right)+\frac{m+1}{2}=2m-1.$$
If $s_v<-1$, then $s_v\leq -3$. Now, by Lemma 1, $s_u\geq 3$ for
each $u\in X$. Therefore we find that, $$\sum_{e\in
E(K_{m,n})}f(e)=\sum_{u\in X}s_u\geq 3m>2m-1.$$

Hence $\gamma'_s(K_{m,n})\geq \min(2m-1, n)$. We now show that there
are two SEDF $f$ and $g$ such that $\sum_{e\in E(K_{m,n})}f(e)=2m-1$
and $\sum_{e\in E(K_{m,n})}g(e)=n$.

Define  $f$ and $g$   as follows,

\ \ $f(u_iv_j) =
\begin{cases}
    1        &\ {\rm if}\,\, i+j\, \,{\rm \, is \, odd} \\
    1         &\ {\rm if}\,\,i=j\\
    -1        &\mbox{otherwise}.
\end{cases}$

It is straightforward to verify that $s_{u_i}=3$, if $i$ is even;
and $s_{u_i}=1$, if $i$ is odd. Also, we have,

\ \ $s_{v_j}=
\begin{cases}
3  &\ {\rm if }\,\, j\,\, {\rm is\,\, even \,\, and}\,\, j \leq m \\
1  &\ {\rm if}\,\, j\,\, {\rm is\,\, odd \,\, and} \,\, j\leq m \\
1  &\ {\rm if}\,\, j \,\,{\rm is\,\, even \,\, and}\,\, j>m \\
-1  &\ {\rm if}\,\, j\,\,{\rm is\,\, odd \,\, and} \,\, j>m.
\end{cases}$

 Consequently, $f$ is an SEDF, by lemma 1. Therefore,
%
$$\gamma'_s(K_{m,n})\leq\sum_{e\in E(K_{m,n})}f(e)=\sum_{u\in X}s_u=3\left(\frac{m-1}{2}\right)+\frac{m+1}{2}=2m-1,$$
as required.

Define $g$ as follows:

 \ \ $g(u_iv_j)=
\begin{cases}
1  &\ {\rm if}\, \, i+j \,\, {\rm is\,\, odd}\\
1        &\ {\rm if}\, \, j\,\,{\rm  is \,\,odd} {\rm \,\, and\,\,}i=j\,\, {\rm modulo \,\,} (m+1) \\
-1 &\ {\rm otherwise}.
\end{cases}$

 It is not hard to see that for any $u\in X$,
$s_u\geq 1$ and for any $v\in Y, s_v=1.$ Therefore $g$ is an SEDF
and $\gamma'_s(K_{m,n})\leq\sum_{e\in E(K_{m,n})}g(e)=\sum_{v \in
Y}s_v=n$.

{\bf (iii)} Three cases may be considered:

\noindent {\bf Case 1.} $n+1\leq 2m$. We claim that
$\gamma'_s(K_{m,n})=2m$. First we show that
$\gamma'_s(K_{m,n})\geq 2m$. By contradiction suppose that there
exists an SEDF, say $f$, such that $\sum_{e\in E(K_{m,n})}f(e)<
2m.$ Since $m\leq n$, we find that $\sum_{e\in E(K_{m,n})}f(e)<
2n$. Thus there exists a vertex $v\in Y$ such that $s_v<2$. On
the other hand since $s_v$ is even, $s_v\leq 0$. If $s_v=0$, then
$v$ is incident with $m/2$ edges with value $1$ and $m/2$ edges
with value $-1$. If $f(uv)=1$, for some $u\in X$, then by Lemma
1, we have, $s_u\geq 2$. Since $s_u$ is odd we find $s_u\geq 3$.
If $f(uv)=-1$, for some $u\in X$, then by a similar argument one
can see that $s_u\geq 1$. Thus,
$$\sum_{e\in E(K_{m,n})}f(e)=\sum_{u\in X} s_u\geq 3m/2+m/2=2m,$$ a
contradiction.  Hence $\gamma'_s(K_{m,n})\geq 2m$.

If $s_v<0$, then $s_v\leq -2$. By Lemma~\ref{edge}, for every
$u\in X$, $s_u\geq 2$. Hence we obtain that, $$\sum_{e\in
E(K_{m,n})}f(e)=\sum_{u\in X}s_u\geq 2m,$$ a contradiction.

We now define an SEDF, say $f$, such that $\sum_{e\in
E(K_{m,n})}f(e)=2m$. Let $X_1=\{u_1, \ldots
,u_{\frac{m}{2}}\}$,$X_2=X-X_1,Y_1=\{v_1,\ldots
,v_{\frac{n+1}{2}}\}$ and $Y_2=Y-Y_1$.

Now, define $f$ as follows:

\ \ $f(e) =
\begin{cases}
    1        &\mbox{if $e$ meets}\,\, X_1\,\, {\rm and}\,\, Y_2 \\
    1        &\mbox{if $e$ meets}\, \,X_2\, \,{\rm and }\, Y_1 \\
    1         &\mbox{if}\, e=u_iv_i,\,\, 1\leq i\leq m/2\\
    1        &\mbox{if}\,e=u_iv_j,\, 1\leq i\leq m/2\,\,{\rm and}\,\, j=(i+m/2) \,\,{\rm modulo}\,\,(n+1)/2 \\
    -1       &\mbox{otherwise}.
\end{cases}$

For each $u\in X_1$, we have $s_u=3$. For every $u\in X_2$, we
have $s_u=1$. Also for each $v\in Y_1$, we have $s_v\geq 2$. For
each $v\in Y_2$, $s_v=0$. By Lemma 1, it is not hard to see that
$f$ is an SEDF. Also we have, $$\sum_{e\in
E(K_{m,n})}f(e)=\sum_{u\in X}s_u=\frac{3m}{2}+\frac{m}{2}=2m.$$

\noindent {\bf Case 2}. $2m<n+1\leq 3m$. We claim that
$\gamma'_s(K_{m,n})=n+1$. First we show that
$\gamma'_s(K_{m,n})\geq n+1$. By contradiction assume that there
exists an SEDF, $f$, such that $\sum_{e\in E(K_{m,n})}f(e)<n+1$.
Since $n+1\leq 3m$, we have $\sum_{e\in E(K_{m,n})}f(e)<3m$.
Therefore there exists a vertex $u\in X$ such that $s_u<3$. Since
$s_u$ is odd, $s_u\leq 1$. If $s_u=1$, then $u$ is incident with
$\frac{n+1}{2}$ edges with value $1$ and $\frac{n-1}{2}$ edges
with value $-1$. If $f(uv)=1$, for some $v\in Y$, then by
Lemma~\ref{edge}, $s_v\geq 1$ and since $s_v$ is even, we have
$s_v\geq 2$. If $f(uv)=-1$, for some $v\in Y$, then one can see
that $s_v\geq 0$. Hence,
$$\sum_{e\in E(K_{m,n})}f(e)=\sum_{v\in Y}s_v\geq 2\left(\frac{n+1}{2}\right)=n+1,$$

which is a contradiction.

If $s_u<1$, then $s_u\leq -1$. By Lemma~\ref{edge}, $s_v\geq 1$, for
each $v\in Y$. Thus, $\sum_{e\in E(K_{m,n})}f(e)=\sum_{v\in
Y}s_v\geq n$. Since the number of edges is even, $\sum_{e\in
E(K_{m,n})}f(e)$ is also even. Now, since $n$ is odd, $\sum_{e\in
E(K_{m,n})}f(e)\geq n+1,$ a contradiction. Hence
$\gamma'_s(K_{m,n})\geq n+1$.

We now define an SEDF, say $f$, such that  $\sum_{e\in
E(K_{m,n})}f(e)=n+1$. Let $X_1=\{u_1, \ldots
,u_{\frac{m}{2}}\}$,$X_2=X-X_1$,$Y_1=\{v_1,\ldots
,v_{\frac{n+1}{2}}\}$ and $Y_2=Y-Y_1$.
Let us define,

\ \ $f(e)=
\begin{cases}
 1  &\ {\rm if} \,\, e \,\,{\rm meets} \,\, X_1 \,\,{\rm and}\,\, Y_2 \\
 1  &\ {\rm if }\,\, e \,\, {\rm meets} \,\, X_2 \,\, {\rm and } \,\, Y_1\\
 1  &\ {\rm if} \,\, e=u_iv_j \,\,{\rm and} \,\, i=j \,\, {\rm modulo} \,\, \frac{m}{2}, \,\, 1\leq i\leq \frac{m}{2}, 1\leq j\leq \frac{n+1}{2}\\
 -1 &\ {\rm otherwise}.
\end{cases}$

It is straightforward to see that for each vertex $u\in X_1$,
$s_u\geq 3$ and for each vertex $u\in X_2$, $s_u=1$. Also, for
each $v\in Y_1, s_v=2$ and for each $v\in Y_2,\, s_v=0$. Thus we
have,
$$\sum_{e\in E(K_{m,n})}f(e)=\sum_{v\in Y}s_v=\frac{2(n+1)}{2}=n+1.$$

 By Lemma \ref{edge}, it can be easily seen that $f$ is an
 SEDF.

\noindent {\bf Case 3}.  $3m<n+1$. We claim that
$\gamma'_s(K_{m,n})=3m$. First we prove that
$\gamma'_s(K_{m,n})\geq 3m$. By contradiction assume that there
exists an SEDF $f$ such that $\gamma'_s(K_{m,n})< 3m$. Hence
there exists a vertex $u\in X$ such that $s_u<3$.  By a similar
method as we saw in the proof of Case $2$, we conclude that
$\sum_{e\in E(K_{m,n})}f(e)\geq n+1$, which contradicts the
inequality $3m< n+1$. Hence $\gamma'_s(K_{m,n})\geq 3m$.

We now define an SEDF, say $f$, such that $\sum_{e\in
E(K_{m,n})}f(e)=3m.$ Consider a partition of $X$ such as $X_1$ and
$X_2$, each of them containing $m/2$ vertices. Also suppose that
$Y_1$, $Y_2$ and $Y_3$ is a partition of $Y$ such that
$|Y_1|=|Y_2|=\frac{n-3}{2}$ and $|Y_3|=3$. We define $f$ as follows:

\ \ $f(e)=
\begin{cases}
-1 &\ {\rm if} \,\, e \,\, {\rm meets} \,\, X_1 \,\, {\rm and} \,\, Y_1 \\
-1 &\ {\rm if} \,\, e \,\,  {\rm meets} \,\, X_2\,\,{\rm and}\,\, Y_2\\
1  &\ {\rm otherwise}.
\end{cases}$

 Now, it can be easily seen that for any $u\in X$, $s_u=3$ and
for any $v\in Y, s_v\geq 0$. By Lemma~\ref{edge}, $f$ is an SEDF.
Also we have,
$$\sum_{e\in E(K_{m,n})}f(e)=\sum_{u\in X}s_u=3m.$$

{\bf (iv)} Three cases may be considered:

\noindent {\bf  Case 1.}  $n\leq 2m$. We claim that $
\gamma'_s(K_{m,n})=2m$. First we show that $\gamma'_s(K_{m,n})\geq
2m$. By contradiction suppose that $f$ is an SEDF  such that
$\sum_{e\in E(K_{m,n})}f(e)< 2m$. Thus, there exists a vertex
$u\in X$ such that $s_u<2$. Since $s_u$ is even,  $s_u\leq 0$. If
$s_u=0$, then $\frac{n}{2}$ edges incident with $u$ have value
$1$ and other $\frac{n}{2}$ edges have value $-1$. If $f(uv)=1$,
for some $v\in Y$, then by Lemma 1, $s_v\geq 2$ and since $s_v$
is odd, we have $s_v\geq 3$. If $f(uv)=-1$, then we have $s_v
\geq 1$. Therefore, $$\sum_{e\in E(K_{m,n})}f(e)=\sum_{v\in
Y}s_v\geq 3n/2+n/2=2n> 2m,$$ a contradiction.

Now, assume that $s_u<0$. Thus $s_u\leq -2$. By Lemma 1, $s_v\geq
2$, for any $v\in Y$. Therefore, $$\sum_{e\in
E(K_{m,n})}f(e)=\sum_{v\in Y}s_v\geq 2n> 2m,$$ a contradiction.
Hence $\gamma'_s(K_{m,n})\geq 2m$.

We now define an SEDF, say $f$, such that $\sum_{e\in
E(K_{m,n})}f(e)=2m.$ We know that all edges of $K_{m,n}$ can be
decomposed into $K_{m,m}$ and $K_{n-m,m}$. Note that $m$ and $n-m$
are odd and $n-m\leq m$. By Part (ii) there exists an SEDF, $g_1$,
for $K_{m,m}$ such that $\sum_{e\in E(K_{m,m})} g_1(e)=m$ and for
each vertex $x$, $s_x=1$. Also there exists an SEDF, say $g_2$, for
$K_{n-m,m}$ such that $\sum_{e\in E(K_{n-m,m})}g_2(e)=m$ and for
every vertex $u\in X$, $s_u=1$ and for other vertex $v$, $s_v\geq
1$. Now, define an SEDF, say $f$, for $K_{m,n}$ such that for each
$e\in E(K_{m,m})$, $f(e)=g_1(e)$ and for every $e\in E(K_{n-m,m})$,
$f(e)=g_2(e)$. Now, for every $u\in X$, we have $s_u=2$ and for each
$v\in Y$, we have $s_v\geq 1$. By Lemma \ref{edge}, $f$ is an SEDF
and moreover we find,
$$\sum_{e\in E(K_{m,n})}f(e)=\sum_{e\in
E(K_{m,m})}g_1(e)+\sum_{e\in E(K_{n-m,m})}g_2(e)=m+m=2m.$$

\noindent {\bf Case 2.} $2m< n\leq 3m-1$. We claim that
$\gamma'_s(K_{m,n})=n$. First we show that $\gamma'_s(K_{m,n})\geq
n$. By contradiction assume that $f$ is an SEDF and $\sum_{e\in
E(K_{m,n})}f(e)< n$. This implies that there exists a vertex
$v\in Y$ such that $s_v<1$. Since $s_v$ is odd, we have $s_v\leq
-1$. If $s_v=-1$, then $v$ is incident with $\frac{m-1}{2}$ edges
with value $1$ and  $\frac{m+1}{2}$ edges with value $-1$. If
$f(uv)=1$, for some $u\in X$, then by Lemma~\ref{edge} , $s_u\geq
3$. Now, since $s_u$ is even, $s_u\geq 4$. If $f(uv)=-1$, then we
conclude that $s_u\geq 2$. Thus,
$$\sum_{e\in E(K_{m,n})}f(e)=\sum_{u\in X}s_u\geq
\frac{4(m-1)}{2}+\frac{2(m+1)}{2}=3m-1\geq n,$$ a contradiction.

If $s_v<-1$, then $s_v\leq -3$. By Lemma~\ref{edge}, for every
$u\in X$, $s_u\geq 3$. Hence we obtain,
$$\sum_{e\in E(K_{m,n})}f(e)=\sum_{u\in X}s_u\geq 3m>n,$$
a contradiction. Hence $\gamma'_s(K_{m,n})\geq n$.

By a similar argument as we did in the Case 1, we may find an
SEDF, say $f$, for $K_{m,n}$ such that $\sum_{e\in
E(K_{m,n})}f(e)=m+(n-m)=n$, as desired.

\noindent{\bf Case 3.} $3m-1<n$. We claim that
$\gamma'_s(K_{m,n})=3m-1$. First we show that
$\gamma'_s(K_{m,n})\geq 3m-1$. By contradiction assume that $f$ is
an SEDF such that $\sum_{e\in E(K_{m,n})}f(e)< 3m-1$. Since
$3m-1<n$, there exists a vertex $v\in Y$ such that $s_v<1$. Now,
by a similar argument as we did in Case 2, one can see that
$\sum_{e\in E(K_{m,n})}f(e)\geq 3m-1$, a contradiction.

We now define an SEDF, say $f$, such that $\sum_{e\in
E(K_{m,n})}f(e)=3m-1$.  Consider a partition of $X$ into two subsets
$X_1$ and $X_2$ such that $|X_1|=\frac{m+1}{2}$ and
$|X_2|=\frac{m-1}{2}$. Also consider a partition of $Y$ such as
$Y_1, Y_2$ and $Y_3$ such that $|Y_1|=\frac{3m+3}{2}$,
$|Y_2|=\frac{n}{2}-2$, $|Y_3|=\frac{n-(3m-1)}{2}$. Let
$X_1=\{u_1,\ldots ,u_{\frac{m+1}{2}}\}, Y_1=\{v_1,\ldots ,
v_{\frac{3m+3}{2}}\}$. Define $f$ as follows:

\ \ $f(e) =
\begin{cases}
    1        &\mbox{if e meets}\, \,X_1 \, \,{\rm and}\,\, Y_2 \\
    1        &\mbox{if e meets}\, \,X_2 \,\,{\rm and}\,\, Y_1 \\
    1        &\mbox{if e meets}\, \,X_2\, \,{\rm and}\,\, Y_3 \\
    1        &\ e=u_iv_j, \, \, 1\leq i\leq \frac{m+1}{2}\,\,\, {\rm and}\,\,  j \in \{3i-2, 3i-1, 3i\} \\
    -1       &\mbox{otherwise}.
\end{cases}$

One can easily see that for any $u\in X_1$, $s_u=2$, and for any
$u\in X_2$, $s_u=4$. Also we have,

\ \  $s_v=
\begin{cases}
1   &\  v \in Y_1 \cup Y_2\\
-1  &\  v \in Y_3.
\end{cases}$

 Now, Lemma \ref{edge} implies that $f$ is an SEDF.

Also, we have, $$\sum_{e\in E(K_{m,n})}f(e)=\sum_{u\in
X}s_u=\frac{2(m+1)}{2}+\frac{4(m-1)}{2}=3m-1.$$ }
\end{proof}

\noindent {\bf Acknowledgment.} The research of the first author was
supported by a grant from IPM (No. 86050212).

{}
\noindent {\small {\sc Saieed Akbari  \quad {\tt
s\_akbari@sharif.edu}}\\  \noindent Institute for Studies
in Theoretical Physics and Mathematics,\\ P. O. Box 19395-5746, Tehran, Iran\\ \noindent Department of Mathematical Sciences\\
Sharif
University of Technology\\P. O. Box 11365-9415, Tehran, Iran.\\
\noindent {\small {\sc Pooya Hatami \quad
{\tt p\_hatami@ce.sharif.edu}}\\
\noindent {\small {\sc Sadegh Bolouki \quad       {\tt
saadegh@ee.sharif.edu}}\\ \noindent {\small {\sc
Milad Siami \quad       {\tt miladsiami@ee.sharif.edu }}\\
\noindent Department of Electrical Engineering \\ Sharif
University of Technology\\ Tehran, Iran.}


\begin{thebibliography}{}

\bibitem{bon} J.A. Bondy and U.S.R. Murty, Graph Theory with Applications,
North-Holland, 1976.




\bibitem{xu1}  B. Xu,
Two classes of edge domination in graphs, Disc. Appl. Math. 154
(2006), No. 10, 1541-1546.



\bibitem{xu2} B. Xu,
On edge domination numbers of graphs, Disc. Math. 294 (2005), No.
3, 311-316.

\bibitem{xu3} B. Xu,
On signed edge domination numbers of graphs, Disc. Math. 239
(2001) 179-189.


\bibitem{ze} B. Zelinka,
On signed edge domination numbers of trees, Math. Bohem. 127
(2002), no. 1, 49-55.

\bibitem{zh} Z. Zhang, B. Xu, Y. Li,  L. Liu,
A note on the lower bounds of signed edge domination number of a
graph, Discrete Math. 195 (1999), No. 1-3, 295-298

\end{thebibliography}
\end{document}